\documentclass[11pt]{elsarticle}
\usepackage{tikz}

\usepackage[utf8]{inputenc}
\usepackage{amsfonts}
\usepackage{amsmath}
\usepackage{amssymb}
\usepackage[english]{babel}
\usepackage[T1]{fontenc}

\newcommand{\gf}{\mathbb{F}_2}

\newtheorem{theorem}{Theorem}
\newtheorem{lemma}{Lemma}
\newtheorem{proposition}{Proposition}
\newtheorem{corollary}{Corollary}
\newproof{proof}{Proof}

\journal{Theoretical Computer Science}

\begin{document}

\title{
Cancellation-Free Circuits in Unbounded and Bounded Depth
\tnoteref{titlelabel}}
\tnotetext[titlelabel]{A preliminary version of this paper
appears in~\cite{DBLP:conf/fct/BoyarF13}.
Both authors are partially supported by the 
Danish Council for Independent Research, Natural Sciences.}

\author{Magnus Gausdal Find\fnref{magnuslabel}}
\ead{magnusgf@imada.sdu.dk}
\fntext[magnuslabel]{Part
of this work was done while visiting
the University of Toronto.}
\author{Joan Boyar\fnref{joanlabel}} 
\ead{joan@imada.sdu.dk}
\fntext[joanlabel]{Part
of this work was done while visiting
the University of Waterloo.}

\address{Department of Mathematics and Computer Science, University of Southern Denmark, Denmark}

\begin{abstract}
We study the notion of ``cancellation-free'' circuits.
This is a restriction of XOR circuits,
but can be considered as being
equivalent to previously studied models of computation.
The notion was coined by Boyar and Peralta in a study of heuristics
for a particular circuit minimization problem.
They asked how large a gap there can be between 
the smallest cancellation-free circuit and the smallest XOR
circuit.
We present a new proof showing that 
the difference can be a factor $\Omega(n/\log^{2}n)$.
Furthermore, our proof holds for circuits of constant depth.
We also study the complexity of computing the Sierpinski matrix
using cancellation-free circuits and give a tight
$\Omega(n\log (n))$ lower bound. 
\end{abstract}

\maketitle

\section{Introduction}
Let $\gf$ be the field of order $2$, and let $\gf^n$
be the $n$-dimensional vector space over $\gf$. For $n\in \mathbb{N}$, we let $[n]=\{1,\ldots ,n \}$.
A Boolean function $f\colon \gf^n \rightarrow \gf^m$ is said
to be linear if there exists a Boolean $m\times n$ matrix $A$
such that $f(\mathbf{x})=A\mathbf{x}$ for every $\mathbf{x}\in \gf^n$.
This is equivalent to saying that $f$ can be computed using
only XOR gates.

An \emph{XOR circuit} (or a \emph{linear circuit}) $C$ is a directed
acyclic graph.
There are $n$ nodes with in-degree $0$, called the \emph{inputs}.
All other nodes have in-degree $2$ and are called \emph{gates}.
There are $m$ nodes which are called the \emph{outputs}; these are
labeled $y_1,\ldots,y_m$.
The value of a gate is the sum
of its two children (addition in $\gf$, denoted $\oplus$).
The circuit $C$, with inputs $\mathbf{x}=(x_1,\ldots ,x_n)$,
\emph{computes} the $m\times n$ matrix $A$
if the output vector computed by $C$,
$\mathbf{y}=(y_1,\ldots , y_m)$,
satisfies $\mathbf{y}=A\mathbf{x}$.
In other words, output $y_i$ is defined by the $i$th row of the matrix.
The \emph{size} of
a circuit $C$, is the number of gates in $C$. 
The \emph{depth}
is the number of gates on a longest directed path from an input to an output.
For simplicity, we will let $m=n$ unless otherwise explicitly stated. For a
matrix $A$, let $|A|$ be the number of nonzero entries in $A$.

\paragraph{Our contributions:}
In this paper we deal with a restriction of XOR circuits called
\emph{cancellation-free} circuits, coined in \cite{boyarcombinationalappear},
where the authors noticed that many heuristics for finding small
XOR circuits always produce cancella\-tion-free XOR circuits.
They asked the question of how large a separation there
can be between these two models.
Recently,
Gashkov and Sergeev \cite{gashkov2011complexity} showed that the work of
Grinchuk and Sergeev \cite{grinchukandsergeev}
implied a separation of
$\Omega\left(\frac{n}{\log^6 n\log\log n} \right)$.
An improved separation of $\Omega\left(\frac{n}{\log^2 n} \right)$
follows from Lemma 4.1 and Lemma 4.2
in \cite{DBLP:journals/siamcomp/Jukna06}, although this implied separation
was not published until recently \cite{juknasergeevSurvey}.
We present an alternative proof of the same separation.
Our proof is based on a different construction and uses communication complexity
in a novel way that might have independent interest. 
Like the separation implied in the work \cite{juknasergeevSurvey}, but unlike the
separations demonstrated in \cite{gashkov2011complexity,separatingnew}, our separation holds
even in the case of constant depth circuits.
We conclude that many heuristics for finding XOR circuits
do not approximate better than a factor of
$\Theta \left( \frac{n}{\log^{2}n} \right)$ of the optimal.
We also study the complexity of computing the Sierpinski matrix
using cancellation-free circuits. We show that the
complexity is exactly $\frac{1}{2}n\log n$. Furthermore,
our proof holds for OR circuits.
As a corollary to this we obtain an explicit matrix where the
smallest OR circuit is a factor $\Theta(\log n)$
larger than the smallest OR circuit for its complement.

We also study the complexity of computing the \emph{Sierpinski matrix} (described later),
and show a tight $\frac{1}{2}n\log n$ lower bound for OR circuits and cancellation-free circuits.
This results follows implicitly from the work of Kennes \cite{DBLP:journals/tsmc/Kennes92}, however
our proof is simpler and more direct. Also we hope that our proof can be strengthened to give an
$\omega (n)$ lower bound for XOR circuits for the Sierpinski matrix.
A similar lower bound was shown independently by Selezneva in \cite{seleznevaProc,seleznevaArticle}.

\section{Cancellation-Free XOR Circuits}
\label{cancelfreelinear}
For XOR circuits, the value computed by every gate is the parity
of a subset of the $n$ variables.
That is, the output of every gate $u$ can be considered
as a vector $\kappa(u)$ in the vector space $\gf^n$,
where $\kappa(u)_i=1$ if and only
if $x_i$ is a term in the parity function computed by the gate $u$.
We call $\kappa(u)$ the \emph{value vector} of $u$, and
for input variables define
$\kappa(x_i)=e^{(i)}$, the unit vector having the
$i$th coordinate $1$ and all others $0$.
It is clear by definition that if a gate $u$ has the two children $w,t$, then
$\kappa(u)=\kappa(w)\oplus~\kappa(t)$, where $\oplus$ denotes coordinate-wise
addition in $\gf$.
We say that an XOR circuit is \emph{cancellation-free} if for
every pair of gates $u,w$ where  $u$ is
an ancestor of $w$, then $\kappa(u)\geq \kappa(w)$,
where $\geq$ denotes the usual coordinate-wise partial order.
These are also called SUM circuits in
\cite{separatingnew,juknasergeevSurvey}.

If this is satisfied, the circuit never
exploits the $\gf$-identity, $a\oplus a=0$, so
things do not ``cancel out'' in the circuit.

Although it is not hard to see that
cancellation-free circuits  is equivalent to addition chains
\cite{DBLP:conf/focs/Pippenger76,DBLP:journals/siamcomp/Pippenger80}
and
``ensemble computations''
\cite{DBLP:books/fm/GareyJ79},
we stick to the term ``cancellation-free'',
since we will think
of it as a special case of XOR circuits. 

For a simple example demonstrating that cancellation-free
circuits indeed are less powerful than general XOR circuits, consider the matrix
\[
 A = \begin{pmatrix}
      1 & 1 & 0 &0 \\
      1 & 1 & 1 &0 \\
      1 & 1 & 1 &1 \\
      0 & 1 & 1 &1 \\
     \end{pmatrix}.
\]
In Figure~\ref{fig:examplefig}, two circuits computing the matrix
$A$ are shown, the circuit on the right uses cancellations, and the circuit on
the left is cancellation-free, and has one gate more.
For this particular matrix, any cancellation-free circuit must use
at least $5$ gates.

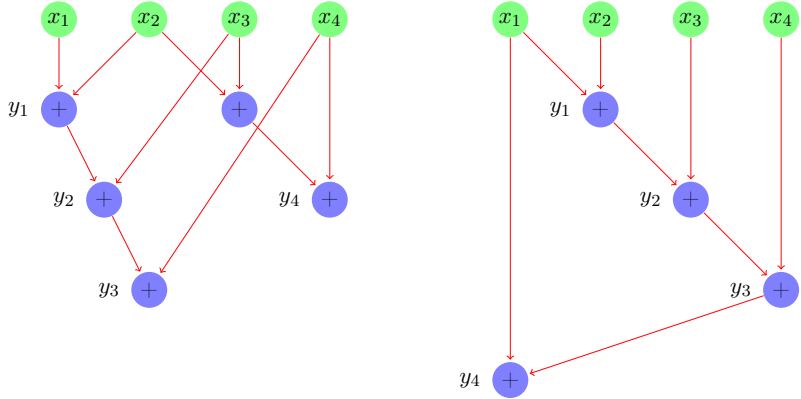
\begin{figure}
\begin{center}

\scalebox{0.8}{
  \begin{tikzpicture}[scale=0.75,shorten >=1pt,->,draw=black!50, node distance=\layersep,invis/.style={inner sep=0pt}]
    \tikzstyle{every pin edge}=[<-,shorten <=1pt]
    \tikzstyle{neuron}=[circle,fill=black!25,minimum size=17pt,inner sep=0pt]
    \tikzstyle{input neuron}=[neuron, fill=green!50];
    \tikzstyle{output neuron}=[neuron, fill=red!50];
    \tikzstyle{hidden neuron}=[neuron, fill=blue!50];
    \tikzstyle{annot} = [text width=4em, text centered]

%
%
%
%
%
%
%
%
%
%
%
%

%
 
  \foreach \name / \x in {1,2,3,4}
          \node[input neuron] (I-\name) at (2*\x,0) {$x_{\x}$};
 
     \node[invis](fisso) at (1,-9){};
 	\node[invis](fisso) at (9,1){}; 	
 
      \node[hidden neuron] (G-1) at (2,-2) {+};
      \path[color=red] (I-1) edge (G-1); 
 	\path[color=red] (I-2) edge (G-1);
     \node[left=0.05cm] at (G-1.west) {$y_1$};

 	\node[hidden neuron] (G-2) at (3,-4) {+};
      \path[color=red] (I-3) edge (G-2); 
 	\path[color=red] (G-1) edge (G-2);
 	\node[left=0.05cm] at (G-2.west) {$y_2$};
 
 	\node[hidden neuron] (G-3) at (4,-6) {+};
      \path[color=red] (I-4) edge (G-3); 
 	\path[color=red] (G-2) edge (G-3);
 	\node[left=0.05cm] at (G-3.west) {$y_3$};

 	\node[hidden neuron] (G-4) at (6,-2) {+};
      \path[color=red] (I-2) edge (G-4); 
 	\path[color=red] (I-3) edge (G-4);

 	\node[hidden neuron] (G-5) at (8,-4) {+};
      \path[color=red] (I-4) edge (G-5); 
 	\path[color=red] (G-4) edge (G-5);
 	\node[left=0.05cm] at (G-5.west) {$y_4$};

    \foreach \name / \x in {1,2,3,4}
         \node[input neuron] (QI-\name) at (2*\x+10,0) {$x_{\x}$};

    \node[invis](Qfisso) at (1,-9){};
	\node[invis](Qfisso) at (9,1){}; 	

     \node[hidden neuron] (QG-1) at (14,-2) {+};
     \path[color=red] (QI-1) edge (QG-1); 
	\path[color=red] (QI-2) edge (QG-1);
    \node[left=0.05cm] at (QG-1.west) {$y_1$};

	\node[hidden neuron] (QG-2) at (16,-4) {+};
     \path[color=red] (QI-3) edge (QG-2); 
	\path[color=red] (QG-1) edge (QG-2);
	\node[left=0.05cm] at (QG-2.west) {$y_2$};

	\node[hidden neuron] (QG-3) at (18,-6) {+};
     \path[color=red] (QI-4) edge (QG-3); 
	\path[color=red] (QG-2) edge (QG-3);
	\node[left=0.05cm] at (QG-3.west) {$y_3$};

	\node[hidden neuron] (QG-4) at (12,-8) {+};
     \path[color=red] (QI-1) edge (QG-4); 
	\path[color=red] (QG-3) edge (QG-4);
	\node[left=0.05cm] at (QG-4.west) {$y_4$};

%
%
%
%
%
%



\end{tikzpicture}
}
\end{center}
\caption{Two circuits computing the matrix $A$. The circuit on the left
is cancellation-free, and has size $5$ - one more more than the circuit
on to the right.\label{fig:examplefig}}
\end{figure}

A different, but tightly related kind of circuits is OR circuits.
The definition is exactly the same as for XOR circuits,
but with $\vee$ (logical OR) instead of $\oplus$, see
\cite{nechiporuk1963rectifier,juknasergeevSurvey,DBLP:books/fm/GareyJ79}.
Cancellation-free circuits is a special case of 
OR circuits and every cancellation-free circuit can be interpreted as an OR circuit for the same matrix, as well as an XOR circuit. 

For a matrix $A$, we will let
$C_\oplus(A)$, $C_{CF}(A)$, $C_\vee(A)$ denote the smallest
XOR circuit, the smallest cancellation-free circuit and the
smallest OR circuit computing the matrix $A$.

By the discussion above, the following is immediate:
\begin{proposition}
\label{orcancelremark}
For every matrix, $A$,  $C_\vee(A)\leq C_{CF}(A)$.
\end{proposition}

This means in particular that any 
lower bound for OR circuits carries over
to a lower bound for cancellation-free circuits.
However, the converse does not hold in general \cite{separatingnew}. A simple
example showing this is the matrix
\[
 B =
 \begin{pmatrix}
  0 & 0 & 1 & 1 & 0 & 0\\
  0 & 1 & 1 & 1 & 0 & 0\\
  1 & 1 & 1 & 1 & 0 & 0\\
  0 & 0 & 1 & 1 & 1 & 0\\
  0 & 0 & 1 & 1 & 1 & 1\\
  1 & 1 & 1 & 1 & 1 & 1\\
 \end{pmatrix}.
\]
For this matrix, there exist an OR circuit with $6$ gates, however
any cancellation-free circuit must have at least $7$ gates.

Every matrix admits a cancellation-free circuit
of size at most $n(n-1)$. This can be obtained simply
by computing each row independently.
It was shown by Nechiporuk \cite{nechiporuk1963rectifier}
and Pippenger
\cite{DBLP:conf/focs/Pippenger76} (see also \cite{juknasergeevSurvey}) that
this upper bound can be improved
to $(1+o(1))\frac{n^2}{2\log n}$.

A Shannon-style counting argument gives that this
is tight up to low order terms. A proof of this can be found in
\cite{DBLP:conf/focs/Pippenger76}.
Combining these results,  we get that for most 
matrices, cancellation does not help much:

\begin{theorem}
For every $0<\epsilon<1$, for sufficiently large $n$, a
random $n\times n$ matrix has
$\frac{C_{CF}(A)}{C_\oplus(A)}\leq 1+\epsilon$ with probability at least
$1-\epsilon$.
\end{theorem}

We also use the following upper bound, which also holds for
cancellation-free circuits, and hence also for OR circuits
and XOR circuits.

\begin{theorem}[Lupanov \cite{Lupanov1956}]
\label{upperboundtranspo}
  Any $m\times n$ matrix, admits a cancellation-free XOR circuit of size
$O\left( \min\{ \frac{mn}{\log n}, \frac{mn}{\log m}\}+n+m\right)$.
\end{theorem}
The theorem follows directly from Lupanov's result and an application of the
``transposition principle'' (see e.g. \cite{juknabook}).

A matrix $A$ is \emph{$k$-free} if it does not have an
all one submatrix of size $(k+1)\times (k+1)$.
The following  lemma will be used later.
According to Jukna and Sergeev \cite{juknasergeevSurvey},
it was independently due to Nechiporuk \cite{nechiporuktopologicalprinciples},
Mehlhorn \cite{mehlhorn1979some},
Pippenger \cite{DBLP:journals/tcs/Pippenger80}, and Wegener \cite{DBLP:journals/acta/Wegener80}.

\begin{lemma}[Nechiporuk, Mehlhorn, Pippenger, Wegener]
\label{freelemma}
For $k$-free $A$, $C_\vee (A) \in \Omega\left(\frac{|A|}{k^2} \right)$.
\end{lemma}

\section[Relationship]{Relationship
Between Cancellation-Free XOR Circuits and General XOR Circuits}
In \cite{boyarcombinationalappear}, Boyar and Peralta exhibited
an infinite family of matrices where the sizes of the cancellation-free circuits computing them
are at least $\frac{3}{2}-o(1)$ times the corresponding sizes for smallest XOR circuits for them. 
We call this ratio the \emph{cancellation ratio}, $\rho(n)$, defined
as

\[
\rho(n) = \max_{A\in \gf^{n\times n}} \frac{C_{CF}(A)}{C_\oplus(A)}.
\]
The following proposition on the Boolean Sylvester-Hadamard matrix
was pointed out by Edward Hirsch
and Olga Melanich \cite{edwardolga}.
The $n\times n$ Boolean Sylvester-Hadamard matrix $H_n$, 
is defined recursively:
\[
H_1 = (1), H_{2n}=
\begin{pmatrix}
  H_{n} & H_n\\
  H_{n} & \overline{H}_n
\end{pmatrix}
\]
Where $\overline{A}$ means the Boolean complement of the matrix
$A$.
It is known that $C_{\oplus}(H_n)\in O(n)$, but that in depth
$2$ it requires circuits of size $\Omega(n\log n)$ \cite{AlonKW90}.

\begin{proposition}
  The $n\times n$ Boolean Sylvester-Hadamard matrix requires
can\-cel\-lation-free circuits of size
$C_{CF}(H_n)\in \Omega(n\log n)$.
\end{proposition}

Since $\log |\det(H_n)| \in \Omega(n\log n)$,
this proposition follows from following theorem due to  
Morgenstern, (\cite{DBLP:journals/jacm/Morgenstern73}, see also
\cite[Thm. 13.14]{DBLP:books/daglib/0090316}).

\begin{theorem}[Morgenstern]\label{morgensternlb}
  For a Boolean matrix  $M$,
\[
C_{CF}\in \Omega(\log |\det(M)|).
\]
\end{theorem}

The statement holds more generally, namely
for circuits with addition
over the complex numbers and scalar multiplication by
any constant $c\in \mathbb{C}$ with $|c|\leq 2$. Cancellation-free
circuits can be seen as a special case of this.

Using the recursive structure of $H_n$, it is not hard to show
that $C_{\oplus}(H_n)\in O(n)$, so this
demonstrates that $\rho(n)\in \Omega(\log n)$.
It should be noted that no $n\times n$ Boolean matrix can have
determinant larger than $n!$, 
so this technique cannot give 
lower a bound on $\rho(n)$ stronger than $O(\log n)$.

As mentioned in the introduction, the ratio 
\[
\lambda (n) = \max_{A\in \gf^{n\times n}}\frac{C_\vee(A)}{C_{\oplus}(A)}
\]
has been studied, (see \cite{gashkov2011complexity,juknasergeevSurvey}).
Using the techniques of \cite{DBLP:journals/siamcomp/Jukna06}, it can
be derived (as is done in \cite{juknasergeevSurvey}) that $\lambda(n)\in \Omega(n/\log^2 n)$.


We present a different construction exhibiting the same gap. The construction
is different, and in some sense simpler. Furthermore our proof is quite different.
More concretely we use communication complexity for the analysis
to show that certain conditional random variables are
almost uniformly distributed in a way
that might have independent interest.
Also our construction gives a similar separation for
circuits of constant depth (see Section \ref{sec:constantdepth}).

\begin{theorem}\label{separationthm}
\(\lambda(n)
\in   
 \Omega\left(\frac{n}{\log ^{2} n} \right).
\)
\end{theorem}

The proof uses the probabilistic method. We construct randomly
two matrices, and let $A$ be their product.
In order to use Lemma \ref{freelemma} on $A$,
we need to show that with high probability, the matrix $A$
will be $2\log n$-free. We do this via Lemma \ref{dependencylemma}
by showing that the marginal distribution of any entry in a fixed $2\log n\times
2\log n$ submatrix is almost uniformly random.

In the following, for a matrix $M$, we let $M_i$ ($M^i$) 
denote its $i$th row (column). And for $I\subseteq [n]$, we let
$M_I$ ($M^I$) denote the submatrix consisting of the rows (columns)
with indices in $I$.

Lemma \ref{dependencylemma} might seem somewhat technical.
However, there is a very simple intuition behind it:
Suppose $M$ is obtained at random as in the statement of the
lemma. Informally we want to say that the entries
do not ``depend'' too much on each other.
More formally we want to show that given 
all but one entry in $M$ it is not possible to guess the last
entry with significant advantage over random guessing.
The proof idea is to transform
any good guess into a deterministic communication protocol for
computation of the inner product, and to use a well known
limitation on how well this can be done
\cite{DBLP:journals/siamcomp/ChorG88,DBLP:books/daglib/0011756}.

We will say that two (partially) defined matrices 
are \emph{consistent} 
if they agree on all their defined entries.

\begin{lemma}
\label{dependencylemma}
Let $M$ be an $m\times m$ partially defined matrix, where
all entries except $M_p^q$ are defined.
Let $B,C$ be matrices over $\gf$ with dimensions
$m\times 7m$ and $7m\times m$ respectively, be uniformly random
among all possible pairs $(B,C)$ such that $BC$ is consistent with
$M$. 

Then for sufficiently large $m$, the conditional probability that $M_{p}^q$ is
$1$, given all other entries, is contained in the interval
$(\frac{1}{2}-\frac{1}{m},\frac{1}{2}+\frac{1}{m})$,
where the probability is over the choices of $B$ and $C$.
\end{lemma}

Before proving the lemma, we will first recall a fact from communication
complexity, due to Chor and Goldreich \cite{DBLP:conf/focs/ChorG85},
see also \cite{DBLP:books/daglib/0011756}.

\begin{theorem}[Chor, Goldreich]
\label{innerproductlb}
  Let $\mathbf{x}$ and $\mathbf{y}$ be independent and
uniformly random vectors, each
of $n$ bits. Suppose a deterministic communication protocol 
is used to compute the inner product of $\mathbf{x}$ and $\mathbf{y}$,
and the protocol is correct with
probability at least $\frac{1}{2}+p$.
Then on some inputs, the protocol uses $\frac{n}{2}-\log( 1/p )$ 
bits of communication. 
\end{theorem}

\begin{proof}[of Lemma \ref{dependencylemma}]
Suppose for the sake of contradiction 
that there exists a partially defined matrix $M$,
such that when all entries
but one are revealed, the conditional probability of the last entry
being $a$ is at least $\frac{1}{2}+\frac{1}{m}$ for some
$a\in \{0,1\}$.

Assuming this, we will first present a randomized communication protocol
computing the inner product of two independent and uniformly random $7m$ bit
vectors $\mathbf{x}$ and $\mathbf{y}$
that always uses $m$ bits of communication and
is correct with probability at least
$\frac{1}{2}+\frac{2^{-2m}}{4m}$.
We will then argue that this
protocol can be completely derandomized.
This results in a
deterministic communication protocol that
violates Theorem \ref{innerproductlb}.
From this we conclude that such a partially defined matrix, 
with this large probability of the last entry being $a$, does not exist.

Let Alice and Bob have as input vectors $\mathbf{x}$ and $\mathbf{y}$,
respectively, each
of length $7m$. 
Before getting their inputs, they use their shared random bits to
agree on a random choice of the two matrices $B$ and $C$ distributed
as stated in the Lemma.
To compute the inner product of $\mathbf{x}$ and $\mathbf{y}$,
Alice replaces the row $B_p$ with
$\mathbf{x}$ and Bob replaces the column $C^q$ with $\mathbf{y}$,
let the resulting matrices be $B'$ and $C'$. Let $M'=B'C'$.
Notice that $M$ and $M'$ are consistent, except possibly on row $p$ and column $q$.
Alice can compute the entire $p$th row of $M'$ (except $(M')_p^q$).
Similarly Bob can compute the entire
$q$th column (except $(M')_p^q$).
 The communication in the protocol consists of first letting Alice send 
the $m-1$ bits in the part
 of the $p$th row she can compute to Bob.  
 Bob now knows all the entries in $M'$, except the entry $M_p^q$.

In order for $M'$ and $M$ to be consistent, it is only necessary that
the $m-1$ defined entries in row $p$ and the $m-1$ defined entries
in column $q$ are equal in the two matrices, since $B'$ and $C'$ were
defined such that all other entries were equal.
This occurs with probability at least
$2^{-2m-2}$.

 In this case, the value
Alice and Bob want to compute is exactly the only unknown entry $M'^q_p$.
By assumption, this last entry is $a$ with probability at least
 $\frac{1}{2}+\frac{1}{m}$, so Bob outputs $a$.
If the known entries in $M'$ are not consistent with the known entries
in $M$, Bob outputs a uniformly random bit.
This is correct with probability $\frac{1}{2}$.
Thus, the probability of this protocol being correct is at least:
\begin{eqnarray*}
&&
2^{-2m-2}\left(\frac{1}{2}+\frac{1}{m}\right)
+
(1-2^{-2m-2})\frac{1}{2}\\
&=&
1/2+\frac{2^{-2m}}{4m}
\end{eqnarray*}

So when the inputs are uniformly distributed, the randomized protocol
computes the inner product of two $7m$ bits vectors with $m$ bits
communication, and it is correct with probability at least
$\frac{1}{2}+\frac{2^{-2m}}{4m}$.
By an  averaging argument it follows that there exist a deterministic
communication protocol with the same success probability.
According to Theorem \ref{innerproductlb}, any deterministic algorithm
for computing the inner product with this success probability
must communicate at least
\[
\frac{7m}{2}-\log(1/p)=\frac{7}{2}m-\log \left( \frac{4m}{2^{-2m}} \right)
=
\frac{3}{2}m -\log m - 2
\] 
Which is larger than $m$ for sufficiently large values of $m$ ($m\geq 16$ suffices), and we arrive at the desired
contradiction.
\qed
\end{proof}

We now use this to prove Theorem \ref{separationthm}.
We will use
following result on the ``Zarankiewicz problem''
\cite{kovari1954problem}, see also \cite{juknacombinatorics}. 
\begin{theorem}[Kov{\'a}ri, S{\'o}s, Tur{\'a}n]
\label{kovarisosturan}
Let $M$ be an $(a-1)$-free $n\times n$ matrix. Then the number of ones in $M$ is 
at most $(a-1)^{1/a}n^{2-1/a}+(a-1)n$.
\end{theorem}

\begin{proof}[of Theorem \ref{separationthm}]
We will probabilistically construct two matrices
$B,C$ of dimensions $n\times 14\log n$,
$14\log n\times n$.
Each entry in $B$ and $C$ will be
chosen independently and uniformly at random on $\gf$.
We let $A=BC$. 
First notice that it
 follows directly from Theorem \ref{upperboundtranspo} that
$B$ and $C$ can be  computed with XOR circuits,
both of size $O(n)$.
 Now we can let the outputs of the circuit computing $C$ be the inputs of
 the circuit computing $B$. Notice that this composed circuit will have many cancellations.
The resulting circuit computes the matrix $A$ and
 has size $O(n)$. We will argue that with probability $1-o(1)$
this matrix will not have a $2\log n\times 2\log n$ submatrix of all ones,
while $|A|\in \Omega(n^2)$.
By Lemma \ref{freelemma} the results follows.

We show that for large enough $n$,
with high probability neither of the following
two events will happen:

\begin{enumerate}
\item \label{all1submatrix}$BC$ has a submatrix of dimension
$2\log n\times 2\log n$ consisting of all ones or all zeros 
\item \label{matrixnorm} $|BC| \leq 0.3n^2$
\end{enumerate}

\paragraph{\ref{all1submatrix})}
Fix a submatrix $M$ of $BC$ with dimensions
$2\log n\times 2\log n$. That is, some subset $I$ of the
rows of $B$, and a subset $J$ of the columns in $C$ so
$M=B_IC^J$. We now want to show that the probability of this matrix
having only ones (or zeros) is so small that a union bound over all
choices of $2\log n\times 2\log n$ submatrices gives that the probability
that there exists such a submatrix goes to $0$.
Notice that this would be easy if all the entries in $M$ were
mutually independent and 
uniformly distributed.

Although this is not case, 
Lemma \ref{dependencylemma} for $m=2\log n$ states,
that this is almost the case. More precisely,
the conditional probability that a given entry is $1$ (or $0$)
is at most 
$\frac{1}{2}+\frac{1}{2\log n}$.
We can now use the union bound to estimate the probability that $A$ has
 a submatrix of dimension $2\log n\times 2\log n$ with all the
entries being either $0$ or $1$:

\begin{eqnarray*}
2{n\choose 2\log n}^2
\left(
\frac{1}{2}+\frac{1}{2\log n}
\right)^{4\log^2n}
&\leq&
2\frac{n^{4\log n}}{(2\log n)!}
\left(
\frac{1+\frac{1}{\log n}}{2}
\right)^{4\log^2n}\\
&\leq&
2\left(\frac{\left(1+\frac{1}{\log n} \right)^{4\log^2 n}}{
(2\log n)!}\right)
\end{eqnarray*}
This tends to $0$, so we arrive at the desired result.

\paragraph{\ref{matrixnorm})}
Note that if one wants to show that with positive probability
the number of ones is $\Omega(n^2)$, a straightforward application of
Markov's inequality suffices.  Here we will show the stronger statement
that with probability $1-o(1)$, the number of ones is
at least $\frac{n^2}{2}-o(n^2)$. By the proof above, we may
assume that the Boolean complement of $A$, $\bar{A}$, does
not have a $2\log n$ submatrix of all ones.
By Theorem \ref{kovarisosturan}, the number of ones in $\bar{A}$
is at most
\[
(2\log n-1)^{1/2\log n}n^{2-1/2\log n}+(2\log n-1)n
\]
One can verify that
\[
\lim_{n\rightarrow \infty}\frac{(2\log n-1)^{1/2\log n}n^{2-1/2\log n}+(2\log n-1)n}{n^2}
=\frac{1}{\sqrt{2}}
\]
So if there is not a $2\log n\times 2\log n$
matrix of all zeros in $A$, the number of zeros in $A$ is at 
most $n^2(1-\frac{1}{\sqrt{2}})<0.3n^2$. Hence the probability of $|A|$ being
less than $0.3{n^2}$ tends to $0$.

\qed
\end{proof}

\emph{Remark 1:}
It has been pointed out by Avishay Tal that in order to show that
the matrix is $O(\log n)$-free, a significantly simpler argument
suffices. We present it here: Let $B,C$ be random matrices as in the
construction of Theorem \ref{separationthm} but with
dimensions $n\times 5\log n$ and $5\log n \times n$, respectively, and
let $A=BC$. 
Now any $5\log n\times 5\log n$ submatrix of $A$ is a product of two
$5\log n\times 5\log n$ dimensional matrices, one being a submatrix
of $B$  and one being a submatrix of $C$. Now recall the
theorem from linear algebra:

\begin{theorem}[Sylvester's Rank Inequality]
  For two $m\times m$ matrices $B,C$
\[
rank(BC) \geq rank(B)+rank(C)-m
\]
\end{theorem}

The probability that a random $k\times k$ matrix has rank less
than $d$ is at most $2^{k-(k-d)^2}$ (see
e.g. the proof of Lemma 5.4 in \cite{DBLP:conf/focs/KomargodskiRT13}).
Now a union bound shows that the probability that there is
a $5\log n\times 5\log n$ submatrix of $B$ or $C$ with
rank smaller than $0.51 \cdot 5\log n$ tends to $0$.
So for large enough $n$, with high probability, every
$5\log n\times 5\log n$ of $A$ will have rank at least
$0.02\cdot 5\log n$.
A submatrix consisting of all ones or all zeros has rank 0 or 1, which
is less than $0.1 \log n$ for large enough $n$. Thus, the
probability of this occurring tends to zero.


In the matrix constructed in \cite[Theorem 5.8]{juknasergeevSurvey}, they
highlight the property that the matrix is 
$t$\emph{-Ramsey}, meaning that both the matrix and its complement are $(t-1)$-free,
and it is a somewhat interesting fact that such
matrices admit small XOR circuits. It follows
immediately from the proof of Theorem \ref{separationthm} that this holds as well for
the matrix constructed, and we state this a separate corollary.

\begin{corollary}
For large enough $n$, 
  with high probability, the bipartite graph with adjacency matrix
$A$ from Theorem~\ref{separationthm} is $t$-Ramsey
for $t=2\log n$. 
\end{corollary}

Notice that by Theorem \ref{upperboundtranspo},
the obtained separation is at most a factor of $O(\log n$) from being
optimal. 
Also, except for lower bounds
based on counting, all strong lower bounds we know of are essentially
based on Lemma \ref{freelemma}.
Following that line of thought,
one might hope to improve the separation above by coming
up with a better choice of $A$ that does not have a
$O(\log^{1-\epsilon}n)\times O(\log^{1-\epsilon}n)$ all $1$ submatrix to get
a stronger lower bound on $C_\vee(A)$, or perhaps hope that a tighter
analysis than the above would give a stronger separation.
However, this direction does not seem
promising. To see this, it follows from Theorem \ref{kovarisosturan}
that a matrix without a
$\log^{1-\epsilon} n\times \log^{1-\epsilon} n$ all $1$ submatrix, the
lower bound obtained using Lemma \ref{freelemma} would be
of order
$O\left(
\frac{n^{2-\frac{1}{\log^{1-\epsilon}n}}}{(\log^{1-\epsilon}n)^2}
\right)$, which
is $o\left(\frac{n^2}{\log^2n} \right)$.

\section{Smallest XOR Circuit Problem}

As mentioned earlier, the notion cancellation-free was introduced by Boyar
and Peralta in \cite{boyarcombinationalappear}.
The paper concerns shortest straight line programs for computing
linear forms, which is equivalent to the model studied in
this paper.
In \cite{DBLP:books/fm/GareyJ79}, it is shown that the
Ensemble Computation Problem (recall that this is equivalent to
cancellation-free)
is $\mathbf{NP}$-complete.
For general XOR circuits, the
problem remains $\mathbf{NP}$-complete \cite{boyarcombinationalappear}.
It was observed in \cite{boyarcombinationalappear}
that several researchers have used heuristics that will always
produce cancellation-free circuits,
see \cite{canright2005very,paar95,SatohMTM01}.
By definition, any heuristic which only produces cancellation-free circuits 
cannot achieve an
approximation ratio better than $\rho(n)$. By Proposition~\ref{orcancelremark},
$\rho(n)\geq \lambda(n)$. Thus, Theorem~\ref{separationthm}  implies that
techniques which only produce cancellation-free circuits are not
guaranteed to be very close to optimal.
\begin{corollary}
  The algorithms in \cite{canright2005very,paar95,SatohMTM01} do
not guarantee approximation ratios better than
$\Theta\left( \frac{n}{\log ^{2}n}\right)$.
\end{corollary}

\section{Constant Depth}

\label{sec:constantdepth}
For unbounded depth, there is no known family of (polynomial time
computable) matrices known to require XOR circuits of 
superlinear size.
However, if one puts restrictions on the depth, superlinear lower
bounds are known  \cite{juknabook}.
In this case, we allow each gate to have unbounded fan-in, and instead
of counting the number of gates we count the number of wires in the
circuit. See Figure~\ref{fig:depth2example} for an example of a depth two circuit.

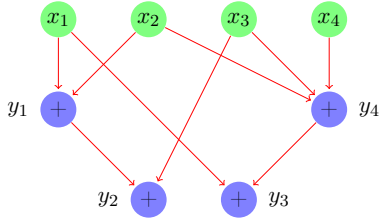
\begin{figure}
\begin{center}
\scalebox{0.8}{
 \begin{tikzpicture}[scale=0.75,shorten >=1pt,->,draw=black!50, node distance=\layersep,invis/.style={inner sep=0pt}]
    \tikzstyle{every pin edge}=[<-,shorten <=1pt]
    \tikzstyle{neuron}=[circle,fill=black!25,minimum size=17pt,inner sep=0pt]
    \tikzstyle{input neuron}=[neuron, fill=green!50];
    \tikzstyle{output neuron}=[neuron, fill=red!50];
    \tikzstyle{hidden neuron}=[neuron, fill=blue!50];
    \tikzstyle{annot} = [text width=4em, text centered]


    \node[invis](fisso) at (1,-9){};
	\node[invis](fisso2) at (9,1){}; 	
    
    \foreach \name / \x in {1,2,3,4}
         \node[input neuron] (I-\name) at (2*\x,0) {$x_{\x}$};

%
%
%
%
%
     \node[hidden neuron] (G-1) at (8,-2) {+};
     \path[color=red] (I-2) edge (G-1); 
	\path[color=red] (I-3) edge (G-1);
	\path[color=red] (I-4) edge (G-1);
    \node[right=0.05cm] at (G-1.east) {$y_4$};

     \node[hidden neuron] (G-2) at (2,-2) {+};
     \path[color=red] (I-1) edge (G-2); 
	\path[color=red] (I-2) edge (G-2);
    \node[left=0.05cm] at (G-2.west) {$y_1$};    

	\node[hidden neuron] (G-3) at (4,-4) {+};
     \path[color=red] (G-2) edge (G-3); 
	\path[color=red] (I-3) edge (G-3);
    \node[left=0.05cm] at (G-3.west) {$y_2$};    

		\node[hidden neuron] (G-4) at (6,-4) {+};
     \path[color=red] (G-1) edge (G-4); 
	\path[color=red] (I-1) edge (G-4);
    \node[right=0.05cm] at (G-4.east) {$y_3$};    

%
%
%
%
%
%



\end{tikzpicture}
}
\end{center}
\caption{An example of a depth $2$ circuit, computing the same matrix as
the circuits in Figure~\ref{fig:examplefig}. Notice that some gates have
fan-in larger than $2$. This circuit has size $9$.
\label{fig:depth2example}
}
\end{figure}

In particular, the circuit model where
the depth is bounded to be at most $2$ is well studied
(see e.g. \cite{juknabook}).
Similarly to previously, an XOR circuit in depth $2$ is a circuit
where
each gate computes
the XOR or its inputs. 
When considering matrices computed by XOR circuits, the general
situation in the two circuit models is very similar.
The following two results are due to Lupanov 
\cite{Lupanov1956}, see also \cite{juknabook}.

\begin{theorem}[Lupanov]
\label{lupanov2}
For every $n\times n$ matrix $A$, there exists
a depth 2 cancel\-lation-free circuit with at most
$O\left(\frac{n^2}{\log n} \right)$ wires computing $A$.
  Furthermore, almost every such matrix requires
$\Omega\left(\frac{n^2}{\log n} \right)$ wires.
\end{theorem}

Let $\lambda^d(n)$ denote $\lambda(n)$ for
circuits restricted to depth $d$ (recall that now size
is defined as the number of wires).
Neither  the separation in
\cite{gashkov2011complexity} nor that in
\cite{separatingnew}
seems to carry over to bounded depth circuits in any obvious way.
The separation presented in \cite[Theorem 5.8]{juknasergeevSurvey}  holds 
for any depth $d\geq 2$.

By inspecting the proof of Theorem~\ref{separationthm}, the upper 
bound on the size of the XOR circuit worked as follows:
First construct a circuit to compute $C$, and then construct a
circuit for $B$ with the outputs
of $C$ as inputs, that is, a circuit for $B$ that comes topologically
after $C$.
To get to an upper bound of  $O(n)$ wires, we use
Theorem \ref{upperboundtranspo}.
By using Theorem~\ref{lupanov2}
twice, we get a depth $4$ circuit of that size.

For depths $d=2$ and $d=3$, one can use arguments
similar to those in given in the proof of \cite[Theorem 5.8]{juknasergeevSurvey}) to show that the
separation still holds in these two cases. We summarize
this in the following theorem.
\begin{theorem}
\label{canceldepth4andup}
  Let $d\geq 2$.
 $\lambda^d(n)\in \Omega\left( \frac{n}{\log^{2}n}\right)$.
\end{theorem}

%

%
%
%
%

\section[Sierpinski]{
Computing the Sierpinski Matrix
}
In this section we prove that the $n\times n$ Sierpinski matrix, $S_n$,
needs $\frac{1}{2}n\log n$ gates when computed by a 
cancellation-free circuit, and that this suffices. The proof
strategy is surprisingly simple, it is essentially gate elimination where
more than one gate is eliminated in each step. Neither
Theorem \ref{morgensternlb} nor
Lemma \ref{freelemma} gives anything nontrivial for this matrix.

As mentioned previously, there is no known
(polynomial time computable) family of
matrices requiring XOR circuits of superlinear size.
However there are simple matrices that are
conjectured to require circuits of size $\Omega (n\log n)$.
One such matrix
is the Sierpinski 
 matrix, (Aaronson, personal communication
 and \cite{cstheorystackexchange}). The $n\times n$ Sierpinski 
(also called \emph{set disjointness}) matrix, $S_n$,
is defined inductively 
\[
S_2 = \begin{pmatrix} 1 & 0 \\ 1 & 1\end{pmatrix},
S_{2n} = \begin{pmatrix} S_n & 0 \\ S_n & S_n\end{pmatrix}
\]
Independently of this conjecture, Jukna and Sergeev~\cite[Problem 7.11]{juknasergeevSurvey}
have very recently asked if the ``set intersection matrix'', $K_n$, has
$C_\oplus (K_n)\in \omega(n)$.
The motivation for this is that 
$C_\vee(K_n)\in O(n)$, so if true
this would give a counterpart to Theorem \ref{separationthm}.

If $n$ is a power of two, the $n\times n$ set intersection matrix $K_n$
can be defined by associating each row and column with a subset of
$[\log n]$, and letting an entry be $1$ if and only if the  corresponding
 row and column sets have non-empty intersection.
One can also define $K_n$ inductively:
\[
K_2 = \begin{pmatrix} 0 & 0 \\ 0 & 1\end{pmatrix},
K_{2n} = \begin{pmatrix} K_n & K_n \\ K_n & J\end{pmatrix},
\]
where $J$ is the $n\times n$ matrix with $1$ in each entry.
It is easy to see that up to a reordering of the columns,
the complement of $K_n$ contains exactly the same rows as $S_n$.
Thus, $C_\oplus(K_n)$ is superlinear
if and only if $C_\oplus(S_n)$ is, since either matrix can be computed
from the other with at most $2n-1$ extra XOR gates, using cancellation
heavily.

To see that the set intersection matrix can be computed with
OR circuits of linear size observe that over the Boolean
semiring, $K_n$ decomposes into  
$K_n = B\cdot B^T$, where the $i$th row in $B$ is the binary
representation of $i$. Now
apply Theorem \ref{upperboundtranspo} to the $n\times \log n$ matrix 
$B$ and its transpose and perform the composition.

Any lower bound against XOR circuits must hold for cancellation-free
circuits, so a first step in proving superlinear lower bounds for the
set intersection matrix is to prove superlinear cancellation-free lower bounds
for the Sierpinski  matrix. Below we show that
$C_{CF}(S_n)= \frac{1}{2}n\log n$.
Our technique also holds for OR circuits.
This provides a simple example of a matrix family where the complements
are significantly easier to compute with OR circuits
than the matrices themselves.

\paragraph{Gate Elimination}

Suppose some subset of the input variables are restricted to the value $0$.
Now look at the resulting circuit.
Some of the gates will now compute the value $z=0\oplus w$.
In this case, we say that the gate is eliminated
since it no longer does any computation. The situation can be more
extreme, some gate might ``compute'' $z=0\oplus 0$.
In both cases, we can remove the gate from the circuit, and
forward the input if necessary (if $z$ is an output
gate, $w$ now outputs the result). In the second
case, the parent of $z$ will get eliminated, so the
effect might cascade.
For any subset of the variables,
there is a unique set of gates that become eliminated when setting
these variables to $0$.

In all of the following let $n$ be a power of $2$, and let $S_n$ be the
$n\times n$ Sierpinski  matrix.
The following proposition is easily established.

\begin{proposition}
For every $n$, the Sierpinski matrix $S_n$ has full rank, over both
 $\mathbb{R}$ and $\mathbb{F}_2$.
\end{proposition}

We now proceed to the proof of the lower bound of
the Sierpinski  matrix for cancellation-free circuits.
It is our hope that this might be a step towards proving
an $\omega(n)$ lower bound for XOR circuits.
\begin{theorem}
\label{sierpinskilower}
For every $n\geq 2$, any cancellation-free circuit that computes 
the $n\times n$ Sierpinski  matrix has size at least
$\frac{1}{2}n\log n$.
\end{theorem}

\begin{proof}
The proof is by induction on $n$. For the base case, look at
the $2\times 2$ matrix $S_2$.
This clearly needs at least
$\frac{1}{2}2\log 2=1$ gate.

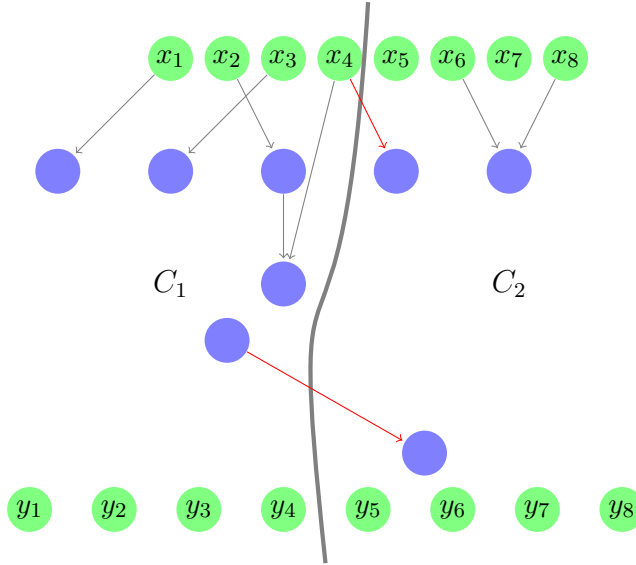
\begin{figure}
  \centering
\begin{tikzpicture}[scale=0.75,shorten >=1pt,->,draw=black!50, node distance=\layersep]
    \tikzstyle{every pin edge}=[<-,shorten <=1pt]
    \tikzstyle{neuron}=[circle,fill=black!25,minimum size=17pt,inner sep=0pt]
    \tikzstyle{input neuron}=[neuron, fill=green!50];
    \tikzstyle{output neuron}=[neuron, fill=red!50];
    \tikzstyle{hidden neuron}=[neuron, fill=blue!50];
    \tikzstyle{annot} = [text width=4em, text centered]


    \foreach \name / \x in {1,2,3,4,5,6,7,8}
         \node[input neuron] (I-\name) at (\x,0) {$x_{\x}$};

    \foreach \y in {1,2,3,4,5,6,7,8}
         \node[input neuron] (O-\y) at (1.5*\y-3,-8) {$y_{\y}$};
    
    \draw[-, ultra thick] (4.5,1) ..  controls (4,-7) and (3,-2) .. (3.75,-9);
    
    \foreach \x/\y in {-1/-2,1/-2,3/-2}
         \node[hidden neuron] (G-\x) at (\x,\y) {};

    \path (I-1) edge (G--1);
    \path (I-2) edge (G-3);
    \path (I-3) edge (G-1);
    
    \node[hidden neuron] (G-5) at (5,-2) {};
    \path[color=red] (I-4) edge (G-5); 

    \node[hidden neuron] (G-Q) at (2,-5) {};

    \node[hidden neuron] (x3x4) at (3,-4) {};
    \path  (I-4) edge (x3x4);
    \path  (G-3) edge (x3x4);
    
    \node[hidden neuron] (G-x2) at (5.5,-7) {};
    \path[color=red] (G-Q) edge (G-x2); 

    \node[font=\large] at (1,-4) {$C_1$};
    \node[font=\large] at (7,-4) {$C_2$};
    
    \node[hidden neuron] (x6x8) at (7,-2) {};
    \path (I-6) edge (x6x8);
    \path (I-8) edge (x6x8);
    




\end{tikzpicture}
  \caption{Figure illustrating the inductive step.
Due to monotinicity there is no wire crossing from right to left.
The gates on the left hand side are in $C_1$.
Notice that wires crossing the cut are red, and that
these wires become constant when $x_1,\ldots ,x_n$ are set
to $0$, so the gates with one such input wire are in $C_3$.
The rest are in $C_2$.}
  \label{fig:sierpinskisit}
\end{figure}

Suppose the statement is true for some $n$ and consider
the $2n\times 2n$  matrix $S_{2n}$. 
Denote the output gates $y_1,\ldots,y_{2n}$
and the inputs $x_1,\ldots,x_{2n}$.
Partition the gates of $C$ into three disjoint sets,
$C_1,C_2$ and $C_3$
(Figure \ref{fig:sierpinskisit} illustrates the situation),
defined as follows:
\begin{itemize}
\item $C_1$: The gates having only inputs from $x_1,\ldots,x_n$ and
$C_1$.
Equivalently the gates not reachable from inputs $x_{n+1},\ldots,x_{2n}$.
\item $C_2$: The gates in $C-C_1$ that are not eliminated
when inputs $x_1,\ldots,x_n$ are set to $0$.
\item $C_3$: $C-(C_1\cup C_2)$. That is, the gates in $C-C_1$ that 
do become eliminated when inputs $x_1,\ldots,x_n$ are set to $0$.
\end{itemize}
Obviously $|C|=|C_1|+|C_2|+|C_3|$.
We will now give lower bounds on the sizes of $C_1$, $C_2$, and $C_3$.

\paragraph{$C_1$:}
Since the circuit is cancellation-free, the
outputs $y_1,\ldots,y_n$ and all their predecessors are in $C_1$.
By the induction hypothesis, $|C_1|\geq \frac{1}{2}n\log n$.

\paragraph{$C_2$:}
Since the gates in $C_2$ 
are not eliminated, they
compute $S_n$ on the
inputs $x_{n+1},\ldots,x_{2n}$. By the induction hypothesis,
 $|C_2|\geq \frac{1}{2}n\log n$.

\paragraph{$C_3$:} The goal is to prove that this set has size
at least $n$. 
Let $\delta(C_1)$ be the set of wires
from $C_1\cup \{x_1,\ldots,x_n\}$ to $C_2\cup C_3$. We first prove that
$|C_3|\geq |\delta(C_1)|$.

By definition, all gates in $C_1$ attain the value $0$ when
$x_1,\ldots,x_n$ are set to $0$.
Let $(v,w)\in \delta(C_1)$ be arbitrary.
Since $v\in C_1\cup \{x_1,\ldots , x_n\}$,
$w$ becomes eliminated, so $w\in C_3$. By definition, every $u\in C_3$
can only have one child in $C_1$.
So $|C_3|\geq |\delta(C_1)|$.

We now show that $|\delta(C_1)|\geq n$.
Let the endpoints of $\delta(C_1)$ in $C_1$ be $e_1,\ldots , e_p$
and let their corresponding value vectors be $v_1,\ldots , v_p$.

The circuit is cancellation\--free, so coordinate\-wise
addition corresponds to addition in $\mathbb{R}$. 
Now look at the value vectors of the output gates $y_{n+1},\ldots,y_{2n}$.
For each of these, the vector consisting of the
first $n$ coordinates must be in $span_{\mathbb{R}}(v_1,\ldots ,v_p)$, but
the dimension of $S_n$ is $n$, so $p\geq n$.
We have that $|C_3|\geq |\delta(C_1)|\geq n$, so
\[
|C|=|C_1|+|C_2|+|C_3| \geq
\frac{1}{2}n\log n + \frac{1}{2}n\log n
+ n = \frac{1}{2}(2n)\log(2n).
\]
\qed
\end{proof}
This is tight:
\begin{proposition}
  The Sierpinski matrix can be computed
by a cancellation-free circuit
using $\frac{1}{2}n\log n$ gates.
\label{constructiveSierpinskiUpperBound}
\end{proposition}
\begin{proof}
This is clearly true for $S_2$. Assume that $S_n$ can be computed using
$\frac{1}{2}n\log n$ gates. Consider the matrix $S_{2n}$.
Construct the circuit
in a divide and conquer manner by constructing recursively on the variables $x_1,\ldots,x_n$
and $x_{n+1},\ldots, x_{2n}$. This gives outputs $y_1,\ldots,y_n$.
After this use $n$ operations to finish the outputs $y_{n+1},\ldots y_{2n}$.
This adds up to exactly $\frac{1}{2}(2n)\log (2n)$. \qed
\end{proof}

\paragraph{Circuits With Cancellation}

In the proof of Theorem \ref{sierpinskilower},
we used
the cancel\-lation-free property when estimating the sizes of both $C_1$
and $C_3$. However,
since $S_n$ has full rank over $\gf$,
a similar dimensionality argument to that used when estimating $C_3$
holds even if the circuits use cancellation.
Therefore we might replace
the cancellation-free assumption with the assumption that for
the $2n\times 2n$
Sierpinski matrix, there is no path from $x_{n+i}$ to
$y_j$ for $i\geq 1$, $j\leq n$.
We have not been able to show whether or not this is the case for
minimum sized circuits, although we have experimentally verified
that even for circuits where cancellation is allowed, 
 the matrices $S_2,S_4,S_8$ do not admit circuits smaller than the
lower bound from Theorem  \ref{sierpinskilower}.

\paragraph{OR circuits}
In the proof of Theorem \ref{sierpinskilower},
the estimates for $C_1$ and $C_2$ hold for OR circuits too,
but when estimating $C_3$, it does not suffice to appeal to rank over
$\gf$ or $\mathbb{R}$.
However, it is not hard to see that any set of row vectors
that ``spans'' $S_n$
(with the operation being coordinate-wise OR)
must have size at least $n$.
\begin{theorem}
  Theorem \ref{sierpinskilower} holds for OR circuits as well.
\end{theorem}

This proof strategy for Theorem~\ref{sierpinskilower} has recently been
used by Sergeev to prove
similar lower bounds for another
family of Boolean matrices  in the OR model \cite{sergeevadditivecompl}.
As mentioned in the introduction, 
Theorem~\ref{sierpinskilower} can be shown using another strategy.
In \cite{DBLP:journals/tsmc/Kennes92}, Kennes gives a
lower bound on the additive complexity for computing the M\"{o}bius
transformation of a Boolean lattice. 
It is not
hard to verify that the Sierpinski matrix corresponds
to the M\"{o}bius transformation induced by the subset lattice. Combining
this  observation with Kennes' result gives the same lower bound.

Since $C_{\vee}(K_n)\in O(n)$ and $K_n$ contains the same rows
as $\bar{S}_n$, the complement of $S_n$,  the Sierpinski matrix
is harder to compute than its complement.

\begin{corollary}
  $C_{\vee}(S_n) = \Theta (\log n) C_\vee(\bar{S}_n)$.
\end{corollary}

Until very recently, this was the largest gap between the OR complexity
of $A$ and $\bar{A}$ for an explicit matrix. See \cite{igorimproves}
for a very recent manuscript describing a construction greatly improving
on this.
\cite{juknasergeevSurvey}).

\section{Conclusions and Open Problems}
We show the existence of matrices, for which 
OR circuits and cancellation-free XOR circuits are both
a factor of $\Omega\left( \frac{n}{\log^{2} n}\right)$ larger
than the smallest XOR circuit. This separation holds
in unbounded depth and any constant depth of at least $2$.

This means that when designing XOR (sub)circuits, 
it can be important that the methods employed can produce
circuits which have cancellation.

If a cancellation-free or an OR circuit computes
the Sierpinski matrix correctly, it has 
size at least $\frac{1}{2}n\log n$. For this particular family of matrices,
it is not obvious to what extent cancellation can help. It would
be very interesting to determine this,
since it would automatically provide a converse
to Theorem~\ref{separationthm}.  

\section*{Acknowledgments}
The authors would like to thank Elad Verbin for an idea which eventually
led to the proof of Theorem~\ref{separationthm},
Igor Sergeev and Stasys Jukna for references to related papers,
Janne Korhonen for pointing out the result of Kennes, Avishay Tal 
for pointing out an alternative proof of a slightly weaker version of
Theorem~\ref{separationthm}, and Mika G\"{o}\"{o}s for 
helpful discussions.

We would also like to thank the anonymnous referees for many valuable
suggestions.

\bibliographystyle{model1-num-names.bst}
\bibliography{bibCancel}

\begin{thebibliography}{36}
\expandafter\ifx\csname natexlab\endcsname\relax\def\natexlab#1{#1}\fi
\providecommand{\url}[1]{\texttt{#1}}
\providecommand{\href}[2]{#2}
\providecommand{\path}[1]{#1}
\providecommand{\DOIprefix}{doi:}
\providecommand{\ArXivprefix}{arXiv:}
\providecommand{\URLprefix}{URL: }
\providecommand{\Pubmedprefix}{pmid:}
\providecommand{\doi}[1]{\href{http://dx.doi.org/#1}{\path{#1}}}
\providecommand{\Pubmed}[1]{\href{pmid:#1}{\path{#1}}}
\providecommand{\bibinfo}[2]{#2}
\ifx\xfnm\relax \def\xfnm[#1]{\unskip,\space#1}\fi
\bibitem[{Boyar and Find(2013)}]{DBLP:conf/fct/BoyarF13}
\bibinfo{author}{J.~Boyar}, \bibinfo{author}{M.~G. Find},
\newblock \bibinfo{title}{Cancellation-free circuits in unbounded and bounded
  depth},
\newblock in: \bibinfo{editor}{L.~Gasieniec}, \bibinfo{editor}{F.~Wolter}
  (Eds.), \bibinfo{booktitle}{FCT}, volume \bibinfo{volume}{8070} of
  \textit{\bibinfo{series}{Lecture Notes in Computer Science}},
  \bibinfo{publisher}{Springer}, \bibinfo{year}{2013}, pp.
  \bibinfo{pages}{159--170}.
\bibitem[{Boyar et~al.(2013)Boyar, Matthews, and
  Peralta}]{boyarcombinationalappear}
\bibinfo{author}{J.~Boyar}, \bibinfo{author}{P.~Matthews},
  \bibinfo{author}{R.~Peralta},
\newblock \bibinfo{title}{Logic minimization techniques with applications to
  cryptology},
\newblock \bibinfo{journal}{J. Cryptology} \bibinfo{volume}{26}
  (\bibinfo{year}{2013}) \bibinfo{pages}{280--312}.
\bibitem[{Gashkov and Sergeev(2011)}]{gashkov2011complexity}
\bibinfo{author}{S.~Gashkov}, \bibinfo{author}{I.~Sergeev},
\newblock \bibinfo{title}{On the complexity of linear {B}oolean operators with
  thin matrices},
\newblock \bibinfo{journal}{Journal of Applied and Industrial Mathematics}
  \bibinfo{volume}{5} (\bibinfo{year}{2011}) \bibinfo{pages}{202--211}.
\bibitem[{Grinchuk and Sergeev(2011)}]{grinchukandsergeev}
\bibinfo{author}{M.~Grinchuk}, \bibinfo{author}{I.~Sergeev},
\newblock \bibinfo{title}{Thin circulant matrices and lower bounds on the
  complexity of some boolean operators},
\newblock \bibinfo{journal}{Discret. Anal. \& Issl. Oper.} \bibinfo{volume}{18}
  (\bibinfo{year}{2011}) \bibinfo{pages}{38–53}.
\bibitem[{Jukna(2006)}]{DBLP:journals/siamcomp/Jukna06}
\bibinfo{author}{S.~Jukna},
\newblock \bibinfo{title}{Disproving the single level conjecture},
\newblock \bibinfo{journal}{SIAM J. Comput.} \bibinfo{volume}{36}
  (\bibinfo{year}{2006}) \bibinfo{pages}{83--98}.
\bibitem[{Jukna and Sergeev(2013)}]{juknasergeevSurvey}
\bibinfo{author}{S.~Jukna}, \bibinfo{author}{I.~Sergeev},
\newblock \bibinfo{title}{Complexity of linear boolean operators},
\newblock \bibinfo{journal}{Foundations and Trends in Theoretical Computer
  Science} \bibinfo{volume}{9} (\bibinfo{year}{2013}) \bibinfo{pages}{1--123}.
\bibitem[{Find et~al.(2013)Find, G\"{o}\"{o}s, Kaski, and
  Korhonen}]{separatingnew}
\bibinfo{author}{M.~Find}, \bibinfo{author}{M.~G\"{o}\"{o}s},
  \bibinfo{author}{P.~Kaski}, \bibinfo{author}{J.~Korhonen},
\newblock \bibinfo{title}{Separating {Or, Sum and XOR Circuits}},
\newblock \bibinfo{journal}{arXiv preprint}  (\bibinfo{year}{2013}).
\bibitem[{Kennes(1992)}]{DBLP:journals/tsmc/Kennes92}
\bibinfo{author}{R.~Kennes},
\newblock \bibinfo{title}{Computational aspects of the mobius transformation of
  graphs},
\newblock \bibinfo{journal}{IEEE Transactions on Systems, Man, and Cybernetics}
  \bibinfo{volume}{22} (\bibinfo{year}{1992}) \bibinfo{pages}{201--223}.
\bibitem[{Selezneva.(2012)}]{seleznevaProc}
\bibinfo{author}{S.~Selezneva.},
\newblock \bibinfo{title}{Lower bound on the complexity of finding polynomials
  of boolean functions in the class of circuits with separated variables},
\newblock in: \bibinfo{booktitle}{Proc. of 11-th Int. Seminar on Discrete Math.
  and Its Appl., Moscow}, \bibinfo{year}{2012}.
\bibitem[{Selezneva(2013)}]{seleznevaArticle}
\bibinfo{author}{S.~Selezneva},
\newblock \bibinfo{title}{Lower bound on the complexity of finding polynomials
  of boolean functions in the class of circuits with separated variables},
\newblock \bibinfo{journal}{Computational Mathematics and Modeling}
  \bibinfo{volume}{24} (\bibinfo{year}{2013}) \bibinfo{pages}{146--152}.
\bibitem[{Pippenger(1976)}]{DBLP:conf/focs/Pippenger76}
\bibinfo{author}{N.~Pippenger},
\newblock \bibinfo{title}{On the evaluation of powers and related problems
  (preliminary version)},
\newblock in: \bibinfo{booktitle}{FOCS}, \bibinfo{publisher}{IEEE Computer
  Society}, \bibinfo{year}{1976}, pp. \bibinfo{pages}{258--263}.
\bibitem[{Pippenger(1980)}]{DBLP:journals/siamcomp/Pippenger80}
\bibinfo{author}{N.~Pippenger},
\newblock \bibinfo{title}{On the evaluation of powers and monomials},
\newblock \bibinfo{journal}{SIAM J. Comput.} \bibinfo{volume}{9}
  (\bibinfo{year}{1980}) \bibinfo{pages}{230--250}.
\bibitem[{Garey and Johnson(1979)}]{DBLP:books/fm/GareyJ79}
\bibinfo{author}{M.~R. Garey}, \bibinfo{author}{D.~S. Johnson},
  \bibinfo{title}{Computers and Intractability: A Guide to the Theory of
  NP-Completeness}, \bibinfo{publisher}{W. H. Freeman}, \bibinfo{year}{1979}.
\bibitem[{Nechiporuk(1963)}]{nechiporuk1963rectifier}
\bibinfo{author}{E.~Nechiporuk},
\newblock \bibinfo{title}{Rectifier networks},
\newblock in: \bibinfo{booktitle}{Soviet Physics Doklady},
  volume~\bibinfo{volume}{8}, \bibinfo{year}{1963}, p.~\bibinfo{pages}{5}.
\bibitem[{Lupanov(1956)}]{Lupanov1956}
\bibinfo{author}{O.~Lupanov},
\newblock \bibinfo{title}{On rectifier and switching-and-rectifier schemes},
\newblock \bibinfo{journal}{Dokl. Akad. 30 Nauk SSSR 111, 1171-1174.}
  (\bibinfo{year}{1956}). \bibinfo{note}{English translation available at
  http://www.thi.informatik.uni-frankfurt.de/\~{}jukna/lupanov56.pdf}.
\bibitem[{Jukna(2012)}]{juknabook}
\bibinfo{author}{S.~Jukna}, \bibinfo{title}{Boolean Function Complexity:
  Advances and Frontiers}, \bibinfo{publisher}{Springer Berlin Heidelberg},
  \bibinfo{year}{2012}.
\bibitem[{Nechiporuk(1969)}]{nechiporuktopologicalprinciples}
\bibinfo{author}{E.~Nechiporuk},
\newblock \bibinfo{title}{On the topological principles of self-correction},
\newblock \bibinfo{journal}{Problemy Kibernetika}  (\bibinfo{year}{1969})
  \bibinfo{pages}{5--102}. \bibinfo{note}{(In Russian)}.
\bibitem[{Mehlhorn(1979)}]{mehlhorn1979some}
\bibinfo{author}{K.~Mehlhorn},
\newblock \bibinfo{title}{Some remarks on \mbox{Boolean} sums},
\newblock \bibinfo{journal}{Acta Informatica} \bibinfo{volume}{12}
  (\bibinfo{year}{1979}) \bibinfo{pages}{371--375}.
\bibitem[{Pippenger(1980)}]{DBLP:journals/tcs/Pippenger80}
\bibinfo{author}{N.~Pippenger},
\newblock \bibinfo{title}{On another boolean matrix},
\newblock \bibinfo{journal}{Theor. Comput. Sci.} \bibinfo{volume}{11}
  (\bibinfo{year}{1980}) \bibinfo{pages}{49--56}.
\bibitem[{Wegener(1980)}]{DBLP:journals/acta/Wegener80}
\bibinfo{author}{I.~Wegener},
\newblock \bibinfo{title}{A new lower bound on the monotone network complexity
  of boolean sums},
\newblock \bibinfo{journal}{Acta Inf.} \bibinfo{volume}{13}
  (\bibinfo{year}{1980}) \bibinfo{pages}{109--114}.
\bibitem[{Hirsch and Melanich(2012)}]{edwardolga}
\bibinfo{author}{E.~Hirsch}, \bibinfo{author}{O.~Melanich},
  \bibinfo{howpublished}{Personal communication}, \bibinfo{year}{2012}.
\bibitem[{Alon et~al.(1990)Alon, Karchmer, and Wigderson}]{AlonKW90}
\bibinfo{author}{N.~Alon}, \bibinfo{author}{M.~Karchmer},
  \bibinfo{author}{A.~Wigderson},
\newblock \bibinfo{title}{Linear circuits over {GF(2)}},
\newblock \bibinfo{journal}{SIAM J. Comput.} \bibinfo{volume}{19}
  (\bibinfo{year}{1990}) \bibinfo{pages}{1064--1067}.
\bibitem[{Morgenstern(1973)}]{DBLP:journals/jacm/Morgenstern73}
\bibinfo{author}{J.~Morgenstern},
\newblock \bibinfo{title}{Note on a lower bound on the linear complexity of the
  fast {F}ourier transform},
\newblock \bibinfo{journal}{J. ACM} \bibinfo{volume}{20} (\bibinfo{year}{1973})
  \bibinfo{pages}{305--306}.
\bibitem[{B{\"u}rgisser et~al.(1997)B{\"u}rgisser, Clausen, and
  Shokrollahi}]{DBLP:books/daglib/0090316}
\bibinfo{author}{P.~B{\"u}rgisser}, \bibinfo{author}{M.~Clausen},
  \bibinfo{author}{M.~A. Shokrollahi}, \bibinfo{title}{Algebraic complexity
  theory}, volume \bibinfo{volume}{315} of \textit{\bibinfo{series}{Grundlehren
  der mathematischen Wissenschaften}}, \bibinfo{publisher}{Springer,
  Heidelberg}, \bibinfo{year}{1997}.
\bibitem[{Chor and Goldreich(1988)}]{DBLP:journals/siamcomp/ChorG88}
\bibinfo{author}{B.~Chor}, \bibinfo{author}{O.~Goldreich},
\newblock \bibinfo{title}{Unbiased bits from sources of weak randomness and
  probabilistic communication complexity},
\newblock \bibinfo{journal}{SIAM J. Comput.} \bibinfo{volume}{17}
  (\bibinfo{year}{1988}) \bibinfo{pages}{230--261}.
\bibitem[{Kushilevitz and Nisan(1997)}]{DBLP:books/daglib/0011756}
\bibinfo{author}{E.~Kushilevitz}, \bibinfo{author}{N.~Nisan},
  \bibinfo{title}{Communication Complexity}, \bibinfo{publisher}{Cambridge
  University Press}, \bibinfo{year}{1997}.
\bibitem[{Chor and Goldreich(1985)}]{DBLP:conf/focs/ChorG85}
\bibinfo{author}{B.~Chor}, \bibinfo{author}{O.~Goldreich},
\newblock \bibinfo{title}{Unbiased bits from sources of weak randomness and
  probabilistic communication complexity (extended abstract)},
\newblock in: \bibinfo{booktitle}{FOCS}, \bibinfo{publisher}{IEEE Computer
  Society}, \bibinfo{year}{1985}, pp. \bibinfo{pages}{429--442}.
\bibitem[{Kov{\'a}ri et~al.(1954)Kov{\'a}ri, S{\'o}s, and
  Tur{\'a}n}]{kovari1954problem}
\bibinfo{author}{T.~Kov{\'a}ri}, \bibinfo{author}{V.~S{\'o}s},
  \bibinfo{author}{P.~Tur{\'a}n},
\newblock \bibinfo{title}{On a problem of {K.} {Zarankiewicz}},
\newblock in: \bibinfo{booktitle}{Colloquium Math}, volume~\bibinfo{volume}{3},
  \bibinfo{year}{1954}, pp. \bibinfo{pages}{50--57}.
\bibitem[{Jukna(2001)}]{juknacombinatorics}
\bibinfo{author}{S.~Jukna}, \bibinfo{title}{Extremal Combinatorics - With
  Applications in Computer Science}, Texts in {T}heoretical {C}omputer
  {S}cience, \bibinfo{publisher}{Springer, Heidelberg}, \bibinfo{year}{2001}.
\bibitem[{Komargodski et~al.(2013)Komargodski, Raz, and
  Tal}]{DBLP:conf/focs/KomargodskiRT13}
\bibinfo{author}{I.~Komargodski}, \bibinfo{author}{R.~Raz},
  \bibinfo{author}{A.~Tal},
\newblock \bibinfo{title}{Improved average-case lower bounds for {DeMorgan}
  formula size},
\newblock in: \bibinfo{booktitle}{FOCS}, \bibinfo{publisher}{IEEE Computer
  Society}, \bibinfo{year}{2013}, pp. \bibinfo{pages}{588--597}.
\bibitem[{Canright(2010)}]{canright2005very}
\bibinfo{author}{D.~Canright},
\newblock \bibinfo{title}{A very compact s-box for {AES}},
\newblock in: \bibinfo{booktitle}{CHES}, volume \bibinfo{volume}{6049} of
  \textit{\bibinfo{series}{LNCS}}, \bibinfo{publisher}{Springer, Heidelberg},
  \bibinfo{year}{2010}.
\bibitem[{Paar(1995)}]{paar95}
\bibinfo{author}{C.~Paar},
\newblock \bibinfo{title}{Some remarks on efficient inversion in finite
  fields},
\newblock in: \bibinfo{editor}{B.~Whistler} (Ed.), \bibinfo{booktitle}{IEEE
  Internatiol Symposium on Information Theory}, volume \bibinfo{volume}{5162}
  of \textit{\bibinfo{series}{LNCS}}, \bibinfo{publisher}{Springer,
  Heidelberg}, \bibinfo{year}{1995}, p.~\bibinfo{pages}{58}.
\bibitem[{Satoh et~al.(2001)Satoh, Morioka, Takano, and Munetoh}]{SatohMTM01}
\bibinfo{author}{A.~Satoh}, \bibinfo{author}{S.~Morioka},
  \bibinfo{author}{K.~Takano}, \bibinfo{author}{S.~Munetoh},
\newblock \bibinfo{title}{A compact {R}ijndael hardware architecture with
  {S}-box optimization},
\newblock in: \bibinfo{editor}{C.~Boyd} (Ed.), \bibinfo{booktitle}{ASIACRYPT},
  volume \bibinfo{volume}{2248} of \textit{\bibinfo{series}{LNCS}},
  \bibinfo{publisher}{Springer, Heidelberg}, \bibinfo{year}{2001}, pp.
  \bibinfo{pages}{239--254}.
\bibitem[{Aaronson(2012)}]{cstheorystackexchange}
\bibinfo{author}{S.~Aaronson}, \bibinfo{title}{Thread on
  cstheory.stackexchange.com\hspace{0pt}},
  \bibinfo{howpublished}{http://cstheory.stackexchange.com/questions/1794/circuit\-lower\-bounds\-over\-arbitrary\-sets\-of\-gates},
  \bibinfo{year}{2012}.
\bibitem[{Sergeev(2012)}]{sergeevadditivecompl}
\bibinfo{author}{I.~Sergeev},
\newblock \bibinfo{title}{On additive complexity of a sequence of matrices},
\newblock \bibinfo{journal}{arXiv preprint}  (\bibinfo{year}{2012}).
\bibitem[{Sergeev(2014)}]{igorimproves}
\bibinfo{author}{I.~Sergeev}, \bibinfo{title}{On the or complexity of matrices
  and their complements},
  \bibinfo{howpublished}{http://lovelace.thi.informatik.uni-frankfurt.de/~jukna/Knizka/comment1.html},
  \bibinfo{year}{2014}.

\end{thebibliography}

\end{document}